\documentclass[11pt]{article}
\usepackage{amsmath,amssymb,amsthm,fullpage}
\usepackage{graphicx,bbm,comment,xcolor}
\usepackage[hidelinks]{hyperref}

\usepackage[ruled, noend, noline]{algorithm2e}

\usepackage{paralist}

\newcommand{\ord}[1]{\ensuremath{{#1}^{\text{th}}}}

\newcommand{\cH}{\ensuremath{\mathcal{H}}}

\newcommand{\cO}{\ensuremath{\mathcal{O}}}
\newcommand{\cP}{\ensuremath{\mathcal{P}}}

\newcommand{\cT}{\ensuremath{\mathcal{T}}}

\newcommand{\davg}{\ensuremath{\overline{d}}}
\newcommand{\Dyes}{\ensuremath{\mathcal{D}^+}}
\newcommand{\Dno}{\ensuremath{\mathcal{D}^-}}
\newcommand{\Gyes}{\ensuremath{G^+}}
\newcommand{\Gno}{\ensuremath{G^-}}

\newcommand{\calD}{{\cal D}}

\newcommand{\eps}{\varepsilon}

\newcommand{\avgwit}{b}

\newcommand{\ceil}[1]{\ensuremath{\left\lceil #1 \right\rceil}}

\newcommand{\E}{\ensuremath{\mathop{{}\mathbb{E}}}}
\newcommand{\poly}{\mathrm{poly}}
\newcommand{\dest}{\ensuremath{\widetilde{d}}}
\newcommand{\mestin}{\ensuremath{\widehat{m}}} 
\newcommand{\destin}{\ensuremath{\widehat{d}}} 
\newcommand{\outdeg}{\ensuremath{d^+}}
\newcommand{\dbot}{\ensuremath{d^\bot}}
\newcommand{\adj}[1]{\mathsf{Adj}(#1)}
\newcommand{\adjp}[1]{\mathsf{Adj}^\star(#1)}

\newcommand{\aspace}{\texttt{\char32}}%

\newtheorem{theorem}{Theorem}[section]
\newtheorem{claim}[theorem]{Claim}
\newtheorem{lemma}[theorem]{Lemma}
\newtheorem{fact}[theorem]{Fact}

\newtheorem{observation}[theorem]{Observation}

\theoremstyle{definition}
\newtheorem{definition}[theorem]{Definition}

\title{Erasure-Resilient Sublinear-Time Graph Algorithms}
\date{}
\author{Amit Levi\footnote{David R. Cheriton School of Computer Science, University of Waterloo. Email: {\sf amit.levi@uwaterloo.ca}. Part of this work was done while the author was visiting Boston University.} \and Ramesh Krishnan S. Pallavoor\footnote{Department of Computer Science, Boston University. Email: {\sf rameshkp@bu.edu, sofya@bu.edu}. The work of these authors was partially supported by NSF award CCF-1909612 and was done in part while the authors were visiting the Simons Institute for the Theory of Computing.} \and Sofya Raskhodnikova\footnotemark[2] \and Nithin Varma\footnote{Department of Computer Science, University of Haifa. Email: {\sf nvarma@bu.edu}. The work of this author was partially supported by ISF grant 497/17
and Israel PBC Fellowship for Outstanding Postdoctoral Researchers from India and China. This work was done in part while the author was a student at Boston University.}}

\begin{document}
\maketitle

\begin{abstract}
We investigate sublinear-time algorithms that take partially erased graphs represented by adjacency lists as input. Our algorithms make degree and neighbor queries to the input graph and work with a specified fraction of adversarial erasures in adjacency entries. We focus on two computational tasks: testing if a graph is connected or $\eps$-far from connected and estimating the average degree. For testing connectedness, we discover a threshold phenomenon: when the fraction of erasures is less than $\eps$, this property can be tested efficiently (in time independent of the size of the graph); when the fraction of erasures is at least $\eps,$ then
a number of queries linear in the size of the graph representation is required.
Our erasure-resilient algorithm (for the special case with no erasures) is an improvement over the previously known algorithm for connectedness in the standard property testing model and has optimal dependence on the proximity parameter $\eps$.
For estimating the average degree, our results provide an ``interpolation'' between the query complexity for this computational task in the model with no erasures in two different settings: with only degree queries, investigated by Feige (SIAM J. Comput. `06), and with degree queries and neighbor queries, investigated by Goldreich and Ron (Random Struct. Algorithms `08) and Eden et al. (ICALP `17).
We conclude with a discussion of our model and open questions raised by our work.
\end{abstract}

\section{Introduction}\label{sec:introduction}
The goal of this work is to model and investigate sublinear-time algorithms that run on graphs with incomplete information. Typically, sublinear-time models assume that algorithms have query or sample access to an input graph. However, this assumption does not accurately reflect reality in some situations.
Consider, for example, the case of a social network where vertices represent individuals and edges represent friendships.
Individuals might want to hide their friendship relations for privacy reasons. 
When input graphs are represented by their adjacency lists, such missing information can be modeled as \emph{erased} entries in the lists.
In this work, we initiate an investigation of sublinear-time algorithms whose inputs are graphs represented by the adjacency lists with some of the entries adversarially erased.

In our erasure-resilient model of sublinear-time graph algorithms, an algorithm gets a parameter $\alpha \in [0,1]$ and query access to the adjacency lists of
a graph with at most an $\alpha$ fraction of the entries in the adjacency lists erased. We call such a graph
 {\em $\alpha$-erased} or, when $\alpha$ is clear from the context, {\em partially erased}.
Algorithms access partially erased graphs via degree and neighbor queries.
The answer to a degree query $v$ is the degree of the vertex $v$.
A neighbor query is of the form $(v,i)$, and the answer is the $i^{\text{th}}$ entry in the adjacency list of $v$.
If the $i^{\text{th}}$ entry is erased\footnote{One can consider a more general model where the degrees of some vertices can also be erased. Our algorithms continue to work in this model, since one can determine the degree of a vertex using $O(\log n)$ neighbor queries (irrespective of whether these queries are made to erased adjacency entries).}, the answer is a special symbol $\bot$. A {\em completion} of a partially erased graph $G$ is a valid graph represented by adjacency lists (with no erasures) that coincide with the adjacency lists of $G$ on all nonerased entries. We formulate our computational tasks in terms of valid completions of partially erased input graphs and analyze the performance of our erasure-resilient algorithms in the {\em worst case} over all $\alpha$-erased graphs.
%
We investigate representative problems from two fundamental classes of computational tasks in our model: graph property testing and estimating a graph parameter.

In the context of graph property testing~\cite{GGR98}, we study the problem of testing whether a partially erased graph is connected. Our model is a generalization of the \emph{general graph model} of Parnas and Ron~\cite{PR02} (which is in turn a generalization of the \emph{bounded degree model} of Goldreich and Ron~\cite{GR02}) to the setting with erasures. A partially erased graph $G$ has property $\cal P$ (in our case, is connected) if there exists a completion of $G$ that has the property. For $\eps \in (0,1)$, such a graph with $m$ edges (more precisely, $2m$ entries in its adjacency lists) is $\eps$-far from $\cal P$ (in our case, from being connected) if every completion of $G$ is different in at least $\eps m$ edges from every graph with the property. The goal of a testing algorithms is to distinguish, with high probability, $\alpha$-erased graphs that have the property from those that are $\eps$-far.
For testing connectedness in our erasure-resilient model, we discover a threshold phenomenon: when the fraction of erasures is less than $\eps$, this property can be tested efficiently (in time independent of the size of the graph); when the fraction of erasures is at least $\eps,$ then
a number of queries linear in the size of the graph is
required to test connectedness.
Additionally, when there are no erasures,
our tester has better query complexity than the best previously known standard tester for connectedness~\cite{PR02,BRY14}, also mentioned in the book on property testing by Goldreich~\cite{GoldreichBook}. Our tester has optimal dependence on $\eps$, as evidenced by a recent lower bound in~\cite{PRV} for this fundamental property.

Next, we study erasure-resilient algorithms for estimating the average degree of a graph.
The problem of estimating the average degree of a graph, in the case with no erasures, was studied by Feige~\cite{Fei06}, Goldreich and Ron~\cite{GR08}, and Eden et al.~\cite{ERS17, ERS}.
Feige designed an algorithm that, for all $\eps > 0$, makes $O(\sqrt{n}/\eps)$ degree queries to an $n$-node graph and outputs, with high probability, an estimate that is within a factor of $2+\eps$ of the average degree. He also showed that to get a 2-approximation, one needs $\Omega(n)$ degree queries.
Goldreich and Ron proved that if an algorithm can make uniformly random neighbor queries (that is, obtain a uniformly random neighbor of a specified vertex) then, for all $\eps > 0$, the average degree can be estimated to within a factor of $1+\eps$ using $O(\sqrt{n}\cdot\poly(\log n, {1}/{\eps}))$ queries. 
Eden et al.\ proved a tighter bound of $O(\sqrt{n} \cdot \log \log n\cdot\poly({1}/{\eps}))$ on the query complexity of this problem and provided a simpler analysis.
We describe an algorithm that estimates the average degree of $\alpha$-erased graphs to within a factor of $1 +\min(2\alpha,1) +\eps$
using $O(\sqrt{n} \cdot \log \log n\cdot\poly({1}/{\eps}))$ queries.
Our result can be thought of as an interpolation between the results in~\cite{Fei06} and~\cite{GR08,ERS17,ERS}. 
In particular, when there are no erasures, that is, when $\alpha =0$, we get a $(1+\eps)$-approximation; when all adjacency entries are erased, and only the degree queries are useful, that is, when $\alpha=1,$
we obtain a $(2 +\eps)$-approximation. We also show that our result cannot be improved significantly: to get a $(1+\alpha)$-approximation, $\Omega(n)$ queries are necessary.

\paragraph{Discussion of our model.} For the case of graph property testing, our model is an adaptation of the erasure-resilient model for testing properties of functions by
Dixit et~al.~\cite{DRTV18}.
Dixit et~al.\
designed erasure-resilient testers for many properties of functions, including monotonicity, the Lipschitz property, and convexity. The conceptual difference between
the two models
is that the adjacency lists representation of a graph cannot be viewed as a function. (This is not the case for the adjacency matrix representation.) For a function, erased entries can be filled in arbitrarily and, as a result, they never contribute to the distance to the property.
For the adjacency lists representation, this is not the case: erasures have to be filled so that the resulting completion is a valid graph.
The restrictions on how they can be filled
may result in some contribution to the distance coming from the erased entries\footnote{Because of this, we make an adjustment to the model of Dixit et~al.~\cite{DRTV18}: we measure the distance to the property as a fraction of the completion representation that needs to be changed, as opposed to the fraction of the nonerased representation that needs to be changed.}.
For example, consider the property of bipartiteness.
Let $B$ be a complete balanced bipartite graph $(U,V;E)$, and let $B'$ be obtained from $B$ by
adding an
erased entry to the adjacency list of every vertex in $U$.
Then, in every completion of $B'$, all formerly erased entries have to be changed to make the graph bipartite.

Furthermore, Dixit et~al.~\cite{DRTV18} gave results only on property testing in the erasure-resilient model. We go beyond property testing in our exploration of erasure-resilient algorithms by considering more general computational tasks.

Finally, our model opens up many new research directions, some of which are discussed in Section~\ref{sec:open}.

\subsection{The Model}
We consider simple undirected graphs $G=(V,E)$ represented by 
adjacency lists, where some entries in the adjacency lists could be adversarially erased (these entries are denoted by $\bot$).
\begin{definition}[$\alpha$-erased graph; completion]
	Let $\alpha \in [0,1]$ be a parameter.
	An {\em $\alpha$-erased graph} on a vertex set $V$
    is a concatenation of the adjacency lists of a simple undirected graph $(V,E)$
with at most an $\alpha$ fraction of all entries (that is, at most $2\alpha |E|$ entries) in the lists erased.
A \emph{completion} of an $\alpha$-erased graph $G$ is
the adjacency lists representation of a simple undirected graph $G'$ that coincides with $G$ on all nonerased entries.
\end{definition}
%

\noindent By definition, every partially erased graph has a completion, because it was obtained by erasing entries in a valid graph.

Given a partially erased graph $G$ over a vertex set $V$, we use $n$ to denote $|V|$ and $m$ to denote the number of edges in any completion of $G$, that is, half the sum of lengths of the adjacency lists of all the vertices in $G$.
The average degree, that is, $2m/n$, is denoted by $\davg$.
For $u \in V$, we use $\adj{u}$ to denote the adjacency list of $u$.
The degree $u$, denoted $\deg(u)$, is the length of $\adj{u}$.

\begin{definition}[Nonerased and half-erased edges] \label{def:half-erased}
	Let $G$ be a partially erased graph over a vertex set $V$. For vertices $u,v \in V$, the set $\{u,v\}$ is a \emph{nonerased edge} in $G$ if $u$ is present in  $\adj{v}$ and vice versa.
	The set $\{u,v\}$ is a \emph{half-erased edge} if $u$ is in $\adj{v}$ but $v$ is not in $\adj{u}$, or vice versa.
\end{definition}

Our algorithms make two types of  queries: \emph{degree queries} and \emph{neighbor queries}.
A degree query specifies a vertex $v$, and the answer is $\deg(v)$.
A neighbor query specifies $(v,i)$, and the answer is the $\ord{i}$ entry in $\adj{v}$.

\begin{definition}[Distance to a property; erasure-resilient property tester]
	Let $\alpha \in [0,1]$, $\eps \in (0,1)$ be parameters.
	An $\alpha$-erased graph $G$ satisfies a property $\mathcal{P}$ if there exists a completion of $G$ that satisfies $\mathcal{P}$.
	An $\alpha$-erased graph $G$ is $\eps$-far from a property $\mathcal{P}$ if every completion $G'$ of $G$ is different in at least $\eps m$ edges from every graph that satisfies $\mathcal{P}$.

	An $\alpha$-erasure-resilient $\eps$-tester for a property $\mathcal{P}$ gets parameters $\alpha \in [0,1],\eps\in (0,1)$ and query access to an $\alpha$-erased graph $G$. The tester accepts, with probability at least $2/3$, if $G$ satisfies $\mathcal{P}$. The tester rejects, with probability at least $2/3$, if $G$ is $\eps$-far from $\mathcal{P}.$
\end{definition}

\subsection{Our Results}
In this section, we state our main results for the erasure-resilient model of sublinear-time algorithms.
\subsubsection{Testing Connectedness}
The problem of testing connectedness in the {\em general graph model} (that we further generalize to the erasure-resilient setting) was studied by Parnas and Ron~\cite{PR02}. The results on this fundamental problem are described in Section~10.2.1 in~\cite{GoldreichBook}. The best tester for this problem to date, due to \cite{BRY14}, had query complexity $O\big(\frac{1}{(\eps \davg)^2} \big)$.

We give two erasure-resilient testers for connectedness: one for small values of $\alpha$ and another for intermediate values of $\alpha$. Both testers work for all\footnote
{For $\eps \geq 2/\davg$, we have $\eps m\geq n$. Then testing connectedness is trivial, since every graph can be made connected by adding at most $n-1$ edges.}
values of the proximity parameter, $\eps$.
We first give a tester that works for all $\alpha < \eps/2$. (This tester is presented in Section~\ref{sec:optimal-conn-tester}.)

\begin{theorem}\label{thm:connectedness-tester-work-investment}
	There exists an $\alpha$-erasure-resilient $\eps$-tester for connectedness of graphs with the average degree $\davg$ that has $O\big(\min\big\{\frac{1}{((\eps - 2\alpha)\davg)^2}, \frac{1}{\eps-2\alpha} \log \frac{1}{(\eps-2\alpha)\davg}\big\}\big)$ query and time complexity and works for every $\eps \in (0,2/\davg)$ and $\alpha \in [0,\eps/2)$. The tester has 1-sided error.
When the average degree $\davg$ of the input graph is unknown, $\alpha$-erasure-resilient $\eps$-testing of connectedness (with 1-sided error) has query and time complexity $O(\frac{1}{\eps-2\alpha}\log\frac{1}{\eps-2\alpha})$.
\end{theorem}

Importantly, when the input adjacency lists have no erasures (i.e., when $\alpha=0$), our tester has better query complexity than the previously known best (standard) tester for connectedness, which was due to \cite{BRY14}. We present a standalone algorithm for this important special case in Appendix~\ref{sec:nonerased-connectedness} for easy reference.
By substituting $\alpha=0$ in Theorem~\ref{thm:connectedness-tester-work-investment}, we get $O\big(\min\big\{\frac{1}{(\eps\davg)^2}, \frac{1}{\eps} \log \frac{1}{\eps\davg}\big\}\big)$ query complexity for the case when $\davg$ is known and $O(\frac{1}{\eps}\log\frac{1}{\eps})$ query complexity when $\davg$ is unknown. For the case with no erasures, the improvement in query complexity as a function of $\eps$ is from $O(\frac 1 {\eps^2})$ to  $O(\frac 1 {\eps} \log \frac 1 {\eps})$. The latter is optimal, as evidenced by an $\Omega(\frac 1 {\eps} \log \frac 1 {\eps})$ lower bound for testing connectedness of graphs of degree 2 in~\cite{PRV}.
We note that Berman et al.~\cite{BRY14}
already proved that testing connectedness of graphs (with no erasures) in the bounded degree graph model of \cite{GR02} has query complexity
$O(\frac{1}{\eps} \log \frac{1}{\eps D})$
where $D$ denotes the degree bound.
Our result shows that the same query complexity (with $D$ replaced by $\davg$) is attainable in the general graph model.

Our first tester looks for small connected components that do not have any erasures.
When $\alpha\in[\eps/2,\eps)$, some $\alpha$-erased graphs that are $\eps$-far from connected may not have any connected component that is free of erasures. Consequently,
our first tester
fails to reject such graphs. We give a different algorithm (presented in Section~\ref{sec:gen-conn-tester}) which works by looking for a subset of vertices that has at most one erasure and gets completed to a unique connected component in every completion of the partially erased graph.
(In the beginning of Section~\ref{sec:gen-conn-tester}, we give an explanation, illustrated by Figure~\ref{fig:2erasures}, of why two erasures in a witness may render it not detectable from a local view obtained by a sublinear algorithm.)

\begin{theorem}\label{thm:connectedness-tester}
	There exists an $\alpha$-erasure-resilient $\eps$-tester for connectedness of graphs with the average degree $\davg$ that has $O\big(\frac{1}{(\eps - \alpha)^2 \cdot \davg} \cdot \min\big\{\frac{1}{(\eps-\alpha) \cdot \davg^2},1 \big\} \big)$ query and time complexity and works for every $\eps \in (0,2/\davg)$ and $\alpha \in [0,\eps)$. The tester has 1-sided error.
\end{theorem}

Finally, we show that when $\alpha \geq \eps$, the task of $\alpha$-erasure-resilient $\eps$-testing of connectedness requires examining a linear portion of the graph representation. That is, we discover a phase transition in the complexity of this problem when the fraction of erasures $\alpha$ reaches the proximity parameter $\eps$.

\begin{theorem} \label{thm:connectedness-lb}
	For all $\eps \in (0,1/7],$ every $\eps$-erasure-resilient $\eps$-tester for connectedness  that makes only degree and neighbor queries requires
a number of queries linear in the size of the graph representation.
\end{theorem}

To prove this theorem, we construct (in Section~\ref{sec:connectedness-lb}) a family of partially erased graphs for which it is hard to distinguish connected graphs from graphs that are far from connected. The average degree of the graphs in our constructions is constant. So, the lower bound for this graph family is $\Omega(n)=\Omega(m)$.


\subsubsection{Estimating the Average Degree}\label{sec:our-results-av-deg}

In Section~\ref{sec:alg-av-degree}, we give an erasure-resilient algorithm for estimating the average degree by generalizing the algorithm of
Eden et al.~\cite{ERS17, ERS}
to work for the case with erasures.
\begin{theorem}\label{thm:degree-ub}
	Let $\alpha \in [0,1]$ and $\eps \in (0, 1/2)$. There exists an algorithm that makes $O(\sqrt{n} \cdot \log\log n \cdot \poly(1/\eps))$ degree queries and uniformly random neighbor queries to an $\alpha$-erased input graph of average degree $\davg\geq 1$ and outputs, with probability at least $2/3$, an estimate $\dest$ satisfying
$(1-\eps) \cdot \davg < \dest < (1 + 2 \min(\alpha, \frac{1}{2}) + \eps) \cdot \davg$.
The running time of the algorithm is the same as its query complexity.
\end{theorem}

For graphs with no erasures, a good estimate of the number of edges gives a good estimate of the average degree.
Feige's algorithm~\cite{Fei06} (that has access only to degree queries) counts some edges twice and gets an estimate of the average degree that is within a factor of $2+\eps$.
Goldreich and Ron~\cite{GR08} and Eden et al.~\cite{ERS17,ERS} avoid the issue of double-counting by
ranking vertices according to their degrees
and estimating, within a factor of $1+\eps$, the number of edges going from lower-ranked to higher-ranked vertices.
These algorithms use degree queries and uniformly random neighbor queries.
Having erasures in the adjacency lists is, in a rough sense, equivalent to not having access to \emph{some} of the neighbor queries.
This results in the additional $2\alpha$ error term in the approximation guarantee.
Consequently, when the fraction of erasures approaches 1/2, all the ``relevant'' entries in the adjacency lists of the input graph could be erased, and we enter the regime of having access only to degree queries.

In Section~\ref{sec:lb-av-degree}, we show that, for any fraction $\alpha\in (0,1]$, estimating the average degree of an $\alpha$-erased graph to within a factor of $(1+\alpha)$ requires $\Omega(n)$ queries.
In other words, the approximation ratio of our erasure-resilient algorithm  for estimating the average degree  cannot be improved significantly.

\begin{theorem} \label{thm:degree-lb}
Let $\alpha \in (0,1]$ be rational. For all $\gamma < \alpha$, at least $\Omega(n)$ queries are necessary for every algorithm that makes degree 
and neighbor queries to an $\alpha$-erased 
graph with the average degree $\davg$ and outputs, with probability at least 2/3, an estimate $\dest \in \left[\davg, (1 + \gamma) \davg \right]$.
\end{theorem}

\subsection{Research Directions and Further Observations}
There are numerous research questions that arise from our work. In Section~\ref{sec:open}, we discuss some of them and also give additional observations about variants of our model. We mention open questions about another (weaker) threshold in erasure-resilient testing of connectedness, about erasure-resilient testing of monotone graph properties, about the relationship between testing with erasures and testing with errors, and about the variant of our model that allows only symmetric erasures.
We show that some of the questions we discuss are open in our model, but easy in the bounded-degree version of our model.

\subsection{Related Work}

Erasure-resilient sublinear-time algorithms, in the context of testing properties of functions, were first investigated by Dixit et al.~\cite{DRTV18}, and further studied by Raskhodnikova et al.~\cite{RRV19}, Pallavoor et al.~\cite{PRW20}, and Ben-Eliezer et al.~\cite{BFLR20}.

Property testing in the general graph model was first studied by Parnas and Ron~\cite{PR02}, who considered a relaxed version of the problem of testing whether the input graph has small diameter.
Kaufman et al.~\cite{KKR04} studied the problem of testing bipartiteness in the general graph model and obtained tight upper and lower bounds on its complexity.

Sublinear-time algorithms for estimating various graph parameters have also received significant attention.
There are sublinear-time algorithms for estimating the weight of a minimum weight spanning tree~\cite{CRT05}, the number of connected components~\cite{CRT05,BKM14}, the average degree~\cite{Fei06,GR08}, the average pairwise distance~\cite{GR08}, moments of the degree distribution~\cite{GRS11,ERS17},
and subgraph counts~\cite{GRS11,ELRS17,ERS18,ER18,ABGP18,AKK19}.

\section{Erasure-Resilient Testing of Connectedness}
\label{sec:er-connectedness}
In this section, we present our results on erasure-resilient testing of connectedness in graphs.

\subsection{\texorpdfstring{An Erasure-Resilient Connectedness Tester for $\alpha<\eps/2$}{An Erasure-Resilient Connectedness Tester for small alpha}}\label{sec:optimal-conn-tester}
In this section, we present our connectedness tester for small $\alpha$ and prove Theorem~\ref{thm:connectedness-tester-work-investment}.
The tester looks for witnesses to 
disconnectedness in the form of connected components with no erasures. It repeatedly performs a breadth first search (BFS) from a random vertex until it finds a witness to disconnectedness or exceeds a specified query budget.

A simple counting argument shows that if a partially erased graph is far from connected then it has many small witnesses to disconnectedness. Moreover, the size of the average witness among them is at most some bound $b$ (that we calculate later). Our tester uses BFS to detect a witness to 
disconnectedness of size at most $b$.

The best tester for connectedness to date, by Berman et al.~\cite{BRY14}, uses a technique called the {\em work investment strategy}.
Specifically, their algorithm repeatedly samples a uniformly random vertex $v$, guesses the size of the witness to disconnectedness $C_{(v)}$ containing $v$, and then performs a BFS from $v$ for $|C_{(v)}|^2$ queries.
Clearly, $|C_{(v)}|^2$ queries are enough to detect $C_{(v)}$.
Using the fact that the expected size of a witness is $b$, they argue that their algorithm has complexity $O(b^2)$.

The new idea in our connectedness tester is to perform the BFS from a uniformly random vertex $v$ for $|C_{(v)}|\cdot \deg(v)/2$ queries.
The expected value of the latter quantity is bounded by $E_{(v)}$, where $E_{(v)}$ denotes the number of edges in the witness containing $v$, and the expectation is over the choice of a uniformly random vertex from $C_{(v)}$.
That is, in expectation, the number of queries that we \emph{invest} into the BFS from $v$ is enough to detect $C_{(v)}$.
We show that, overall, the expected complexity of this algorithm is $\widetilde{O}(b\cdot \davg)$, which is smaller than $O(\avgwit^2)$ when $\avgwit > \davg$.

Our erasure-resilient tester is Algorithm~\ref{alg:connectedness-work-investment}, with a small standard modification to ensure that the stated complexity bounds hold in the worst case (not just in expectation).
It is obtained by running the algorithm of Berman et al.\ (generalized to handle erasures) when $b < \davg$ and running the above algorithm otherwise.

Before stating the algorithm, we formalize the notion of the witness to disconnectedness and argue that partially erased graphs that are far from being connected have many witnesses to disconnectedness.

\begin{definition}
[Witness to disconnectedness]
\label{def:conn-comp}
A set $C$ of vertices is a witness to disconnectedness 
in a partially erased graph $G$ if the adjacency lists of vertices in $C$ have no erasures, and $C$ forms a connected component in every completion of $G$.
\end{definition}

\begin{observation} \label{obs:eps-far-cc}
	Let 
	$\eps \in (0,2/\davg)$ and
	$G'$ be an $m$-edge graph (with no erasures) that is $\eps$-far from connected. Then  $G'$ has at least $\eps m + 1$ connected components.
\end{observation}

Next, in Claim~\ref{clm:no-erasures-witness}, we argue that if the fraction of erasures is \emph{small}, \emph{many} of the connected components present in a completion $G'$ are also present as witnesses to disconnectedness 
in $G$.

\begin{claim}\label{clm:no-erasures-witness}
Let $\eps \in (0, 2/\davg)$ and $\alpha \in [0, {\eps}/{2})$. The number of witnesses to disconnectedness 
in an $\alpha$-erased graph $G$ that is $\eps$-far from connected is at least $(\eps-2\alpha)m$.
\end{claim}
\begin{proof}
By Observation~\ref{obs:eps-far-cc}, every completion $G'$ of $G$ has at least $\eps m + 1$ connected components.
The number of connected components in $G'$ with at least one erased entry in the union of its adjacency lists (with respect to $G$) is at most $2\alpha m$.
Hence, the number of
connected components in $G'$ that do not have any erased entry in the union of its adjacency lists (with respect to $G$) is at least
$\eps m - 2\alpha m = (\eps - 2\alpha) m$.
The claim follows.
\end{proof}

Let $\avgwit = 2/((\eps-2\alpha)\cdot \davg)$. By Claim~\ref{clm:no-erasures-witness}, the size of the average witness to disconnectedness is at most $b$. Now we are ready to state Algorithm~\ref{alg:connectedness-work-investment}.

\begin{algorithm}
	\caption{Erasure-Resilient Connectedness Tester for $\alpha<\eps/2$}\label{alg:connectedness-work-investment}
	\SetKwInOut{Input}{input}\SetKwInOut{Output}{output}
	\SetKwFor{RepeatTimes}{repeat}{times}{end}
	\Input{The average degree $\davg,$ parameters $\eps \in (0,{2}/{\davg}), \alpha \in [0,{\eps}/{2})$; query access to an $\alpha$-erased graph $G$} 
	\DontPrintSemicolon
	\BlankLine
		\nl Let 
        $\avgwit \gets 2/((\eps-2\alpha)\cdot \davg)$. \tcp*{the average size of a witness is at most~$\avgwit$}
		\nl \For{$i \in \left[\ceil{\log(4\avgwit)}\right]$}{
		\nl \RepeatTimes{$\ceil{\frac{4\avgwit\ln 6}{2^i}}$}{
		\nl Sample a vertex $v$ uniformly and independently at random. \label{step:vertex-sample-erased}\;
		\nl \eIf{$ \avgwit \leq \davg \log \avgwit$}
		{
		\nl Run a BFS from $v$ until it encounters an erased entry or $(2^i+1)$ vertices.\label{step:bfs-bucketing}\;
		}{
		\nl Query $\deg(v)$;\;
        \nl Run a BFS from $v$ until it encounters an erased entry or $(2^{i-1}\cdot \deg(v)+1)$ edges.\label{step:alternate-bfs-bucketing}
		}
		\nl \lIf{{the BFS 
            explored an entire connected component and didn't encounter an erasure}}{
		 \textbf{reject}.\label{stp:witness-check-bucketing}
		}
		}
		}
		\nl \textbf{Accept}.\;
	
\end{algorithm}

Clearly, Algorithm~\ref{alg:connectedness-work-investment} accepts all connected partially erased graphs.

\begin{lemma}\label{lem:connectedness-work-investment-correctness}
Let $\eps \in (0,2/\davg)$ and $\alpha \in [0,\eps/2)$. Let $G$ be an $\alpha$-erased graph that is $\eps$-far from connected. 
Then Algorithm~\ref{alg:connectedness-work-investment} rejects $G$ with probability at least 5/6.
\end{lemma}
\begin{proof}
Let $V$ be the vertex set of $G$.
We start by defining the quality of a vertex $v\in V$. The definition is different for the two cases, corresponding to the two stopping conditions Algorithm~\ref{alg:connectedness-work-investment} uses for BFS.
First, we consider the case when $\avgwit \leq \davg \cdot \log \avgwit,$ that is,
when Algorithm~\ref{alg:connectedness-work-investment} runs the version of BFS specified in Step~\ref{step:bfs-bucketing}.

\begin{definition}[Quality of a vertex when $\avgwit \leq \davg \cdot \log \avgwit$]\label{def:quality1}
The quality of a vertex $v$, denoted $q(v)$, is defined as follows. If $v$ belongs to a witness to disconnectedness in $G$ then $q(v) = 1/|C_{(v)}|$, where $C_{(v)}$ denotes the witness to disconnectedness that $v$ belongs to.
Otherwise, $q(v) = 0$.
\end{definition}

\noindent The important feature of $q(v)$ is that, for a witness $C$ to disconnectedness, $\sum_{v \in C} q(v) = 1$.

Next, we define the quality of a vertex for the case  when $\avgwit > \davg \cdot \log \avgwit$, that is,
when Algorithm~\ref{alg:connectedness-work-investment} runs the version of BFS specified in  Step~\ref{step:alternate-bfs-bucketing}.
\begin{definition}[Quality of a vertex when $\avgwit > \davg \cdot \log \avgwit$]\label{def:quality2}
Fix a completion $G'$ of $G$.
For a vertex $v \in V$, let $C_{(v)}$ denote the connected component (in $G'$) containing $v$, and let $E_{(v)}$ 
denote the number of edges in $C_{(v)}$.
The quality of a vertex $v$, denoted $q(v)$, is defined as
\begin{align*}
q(v) =
\begin{cases}
0 & \text{ if } C_{(v)} \text{ contains at least one erased entry in } G, \\
\frac{\deg(v)}{2 E_{(v)}} & \text{ if } E_{(v)} > 0,\\
1 & \text{ if } E_{(v)} = 0.
\end{cases}
\end{align*}
\end{definition}
\noindent As was the case for $q(v)$ from Definition~\ref{def:quality1}, for a witness $C$ to disconnectedness, $\sum_{v \in C} q(v) = 1$.


The rest of the proof of Lemma~\ref{lem:connectedness-work-investment-correctness} is the same for both cases.
We analyze the expected quality of 
a uniformly random vertex $v \in V$. Using 
the fact that $\sum_{v \in C} q(v) = 1$ and Claim~\ref{clm:no-erasures-witness},
\[
	\E_{v \in V}[q(v)] = \frac 1 n \sum_{v \in V} {q(v)}
	=\frac 1 n \sum_{\substack{C: C~\text{is a witness} \\ \text{to disconnectedness}}} 1
	\ge \frac{(\eps - 2\alpha)m}{n} = \frac{1}{\avgwit}.
\]

Finally, we apply the following work investment strategy lemma due to~\cite[Lemma 2.5]{BRY14}.

\begin{lemma}[\cite{BRY14}]\label{lem:work-investment}
	Let $X$ be a random variable that takes values in $[0,1]$. Suppose $\E[X] \ge \beta$, and let $t = \lceil \log (4/\beta) \rceil$. For all $i \in [t]$, let $p_i = \Pr[X \ge 2^{-i}]$ and $k_i = \frac{4 \ln 6}{2^i \beta}$. Then
$
		\prod_{i=1}^{t} (1 - p_i)^{k_i} \le \frac{1}{6}.
$
\end{lemma}

We apply Lemma~\ref{lem:work-investment} with $X$ equal to $q(v)$ for a uniformly random $v\in V$.
Set $\beta = 1/\avgwit$ and
$t = \ceil{ \log(4/\beta)}$. 
For $i \in [t]$, set $p_i$ to be the probability that a
vertex $v$ sampled uniformly at random belongs to
a witness to disconnectedness of $G$ that has at most
(i)  $2^i$ vertices, when $\avgwit \leq \davg \cdot \log \avgwit$;
(ii) $2^{i-1} \cdot \deg(v)$ edges, otherwise.
That is, $p_i=\Pr[X \ge 2^{-i}]$.
Similarly, for $i \in [t]$, let $k_i=\frac{4\ln 6}{2^i \beta}$. 
Then the probability that Step~\ref{stp:witness-check-bucketing} of the tester does not reject is $\prod_{i=1}^{t} (1 - p_i)^{k_i}$.
By Lemma~\ref{lem:work-investment}, this step rejects with probability at least $5/6$.
\end{proof}

\begin{proof}[Proof of Theorem~\ref{thm:connectedness-tester-work-investment}]
We start by analyzing the query and time complexity of Algorithm~\ref{alg:connectedness-work-investment}.

\noindent\textbf{Case 1:} When $\avgwit \leq \davg \cdot \log \avgwit,$ the query and time complexity of Algorithm~\ref{alg:connectedness-work-investment} is
\[
	\sum_{i \in \left[\ceil{\log(4\avgwit) }\right]} \ceil{\frac{4\avgwit\ln 6}{2^i}} \cdot 2^{2i} = O\left(\avgwit^2 \right) = O( \min\{ \avgwit^2, \avgwit \davg \cdot \log \avgwit \}).
\]

\noindent\textbf{Case 2:}
When $\avgwit >\davg \cdot \log \avgwit,$ the expected query and time complexity of Algorithm~\ref{alg:connectedness-work-investment} is
$$
\sum_{i \in \left[\ceil{\log 4\avgwit }\right]} \ceil{\frac{4\avgwit\ln 6}{2^i}} \cdot 2^i \cdot \E_{s \in V}[\deg(s)] = O(\avgwit \davg \log \avgwit) = O( \min\{ \avgwit^2, \avgwit \davg \cdot \log \avgwit \}).$$
Substituting the value of $\avgwit$, we get: $O( \min\{ \avgwit^2, \avgwit \davg \cdot \log \avgwit \})
= O\big(\min\big\{\frac{1}{((\eps - 2\alpha)\davg)^2}, \frac{1}{\eps-2\alpha} \log \frac{1}{(\eps-2\alpha)\davg}\big\}\big)$.
The final tester is obtained by running Algorithm~\ref{alg:connectedness-work-investment} and then aborting and accepting if the number of queries exceeds six times its expectation. The final tester then has the query complexity and the running time stated in Theorem~\ref{thm:connectedness-tester-work-investment}.

The final tester never rejects a connected partially erased graph. However, a partially erased graph that is $\eps$-far from connected can get accepted incorrectly if Algorithm~\ref{alg:connectedness-work-investment} accepts it or if the final algorithm aborts. The probability of the former event is at most 1/6, by Lemma~\ref{lem:connectedness-work-investment-correctness}. The probability of aborting is also at most 1/6, by Markov's inequality. By a union bound, the final algorithm accepts incorrectly with probability at most 1/3, completing the proof of the theorem for the case when $\davg$ is given to the algorithm.

We can adjust the algorithm to work without access to the average degree at a small cost in query and time complexity, using the technique explained in Appendix~\ref{sec:unknown-avg-degree}.
\end{proof}

\subsection{\texorpdfstring{Our Erasure-Resilient Connectedness Tester for $\alpha \in [\eps/2, \eps)$}{Our Erasure-Resilient Connectedness Tester for intermediate alpha}}\label{sec:gen-conn-tester}
In this section, we prove Theorem~\ref{thm:connectedness-tester}. We describe and analyze 
a 1-sided error $\alpha$-erasure-resilient $\eps$-tester for connectedness that can work with more erasures in the input graph than Algorithm~\ref{alg:connectedness-work-investment} can handle.
Specifically, the tester 
works for all $\alpha < \eps$. However, it has better performance than Algorithm~\ref{alg:connectedness-work-investment} only for $\alpha \in [\eps/2, \eps)$.

When $\alpha > \eps/2$, an $\alpha$-erased graph that is $\eps$-far from being connected may not contain any witnesses to disconnectedness as defined in Section~\ref{sec:optimal-conn-tester}.
Specifically, every set $C$ of nodes that gets completed to a connected component could have an erasure in the union of the adjacency lists of the nodes in $C$.
To get around this issue, our tester looks for a \emph{generalized witness to disconnectedness}, which is, intuitively, a connected component with at most one erasure.
Observe that a component with two erasures could
have a unique completion,
but impossible to certify as a separate connected component from the local view from any of its vertices. Figure~\ref{fig:2erasures} shows an example of a small component, where a BFS from any vertex will be unable to certify that the graph is disconnected.

\begin{figure}
\begin{minipage}{0.475\textwidth}
\centering
\includegraphics[scale=0.45]{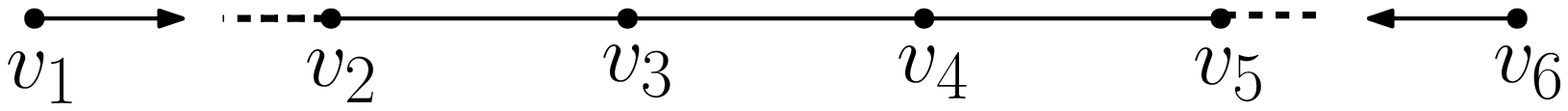}
\caption{\label{fig:2erasures}An example of a component with two erasures, where a BFS from any vertex fails to detect that this component is disconnected from the rest of the graph.}
\end{minipage}\hfill
\begin{minipage}{0.475\textwidth}
\centering
\includegraphics[scale=0.45]{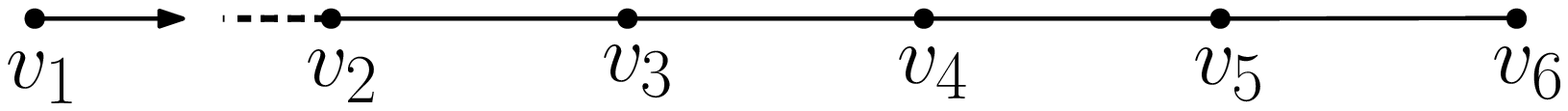}
\caption{\label{fig:1erasure}An example of a generalized witness to disconnectedness, where only a BFS from $v_1$ (but not from any other vertex) detects the generalized witness.}
\end{minipage}

\medskip

A dotted line represents an erasure in the adjacency list of the corresponding vertex.
An arrow pointing from a vertex $a$ in the direction of a vertex $b$ represents that $b\in\adj{a}$, but $a\notin\adj{b}$.
\end{figure}

Our tester repeatedly performs a BFS from a random vertex until it detects a generalized witness to disconnectedness, or exceeds a specified query budget.
We show, by a counting argument, that every partially erased graph that is far from connected has several \emph{small} generalized witnesses to disconnectedness.
The correctness of the tester is ensured by the observation that each such witness $C$ contains at least one vertex from which all the other vertices in $C$ are reachable.
(It is possible to have {\em exactly} one vertex in $C$ from which all the other vertices are reachable. Figure~\ref{fig:1erasure} 
shows an example of a connected component, where a BFS can detect the generalized witness to disconnectedness only if started at vertex $v_1$, but will fail to do so from all other vertices.)

Before we state our tester, we formalize the notion of generalized witnesses.

\begin{definition}[Generalized witness to disconnectedness]\label{def:witness-disconnectedness}
	Given a partially erased graph $G$ over a vertex set $V$, a set $C \subset V$ is a \emph{generalized witness to disconnectedness} of $G$ if
	\begin{enumerate}
 		\item \label{condn:1} there is at most one erased entry ($\bot$) in $\bigcup_{v\in C}\adj{v}$,
		\item\label{condn:2} every nonerased entry in $\bigcup_{v\in C}\adj{v}$ is a vertex from $C$,
		\item\label{condn:3} if $\bot\in\adj{u}$ for some $u \in C$ then $u\in\adj{v}$ but $v\notin\adj{u}$ for some $v \in C$; moreover, each node in $C$ is reachable via a BFS from $v$.
	\end{enumerate}
\end{definition}

Definition~\ref{def:witness-disconnectedness} implies that the only erasure, if any, in the union of the adjacency lists of the nodes in $C$ is part of a half-erased edge within $C,$ and that $C$ forms a connected component in every completion of $G$.

Let $\avgwit = 4/((\eps - \alpha)\davg)$.
Our tester is presented in Algorithm~\ref{alg:connectedness}.
In the rest of the section, we analyze the correctness and complexity of the tester.

\begin{algorithm}
\caption{Erasure-Resilient Connectedness Tester for $\alpha \in [\eps/2, \eps)$}\label{alg:connectedness}
	\SetKwInOut{Input}{input}\SetKwInOut{Output}{output}
	\SetKwFor{RepeatTimes}{repeat}{times}{end}
		\Input{The average degree $\davg,$ parameters $\eps \in (0,{2}/{\davg}), \alpha \in [0,\eps)$; query access to an $\alpha$-erased graph $G$}
		\DontPrintSemicolon
		\BlankLine

		\nl Let $b \gets 4/((\eps - \alpha)\davg)$. \;
		\nl\label{stp:for-loop-connectedness} \RepeatTimes{$\lceil b \ln 3 \rceil$}{
		\nl Sample a vertex $s$ uniformly and independently at random. \label{stp:sampling-bucketing-connectedness} \;
		\nl Run a BFS starting from $s$ using at most $\min\{b^2, b\cdot \davg\}$ neighbor queries.\label{step:bfs} \;
		\nl \label{step:witness-check}\If {Step~\ref{step:bfs} detected a generalized witness to disconnectedness}{
		\nl \textbf{Reject}.
		}
		} \label{stp:end-for-loop-connectedness}
		\nl \textbf{Accept}.\;
\end{algorithm}

\begin{definition}[Small and big sets]\label{def:small-big}
Let $G$ be a partially erased graph and let $\eps^\star \in (0,2/\davg)$ be a parameter. The \emph{representation length} of a set $C$ of nodes is the sum of lengths of the adjacency lists of nodes in $C$. The set $C$ is
 $\eps^\star$-{\em small} if either
	\begin{itemize}
		\item $\eps^\star \geq 4/\davg^2$ and $C$ contains at most $4/( \eps^\star\cdot\davg)$
		vertices, or
		\item $\eps^\star < 4/\davg^2$ and $C$ has representation length at most $4/\eps^\star.$
	\end{itemize}
	The set $C$ is $\eps^\star$-\emph{big} otherwise.
\end{definition}

Claim~\ref{clm:eps-far-charac-2} shows that a partially erased graph that is far from connected has sufficiently many small generalized witnesses to disconnectedness.

\begin{claim}\label{clm:eps-far-charac-2}
	Let $\eps \in (0,2/\davg), \alpha \in [0,\eps)$.
	Let $G$ be an $\alpha$-erased graph that is $\eps$-far from connected.
	The number of $(\eps - \alpha)$-small generalized witnesses to disconnectedness of $G$ is at least $(\eps - \alpha) m/2$.
\end{claim}

\begin{proof}
	We first argue that there are many small connected components in every completion $G'$ of $G$ and then prove that many of these are generalized witnesses in $G$.

	Consider a completion $G'$ of $G$.
	If $\eps - \alpha \geq 4/\davg^2$, the number of $(\eps - \alpha)$-big connected components in $G'$ is 
	at most $n/b = (\eps - \alpha)m/2$. 
	If $\eps - \alpha < 4/\davg^2$, the number of $(\eps - \alpha)$-big connected components in $G'$ is 
	at most $2m/(b\cdot \davg) = (\eps - \alpha)m/2$, since the representation length of the vertex set $V$ of $G$ is $2m$. 
	By Observation~\ref{obs:eps-far-cc}, the total number of connected components in $G'$ is at least $\eps m + 1$.
	Hence, the number of $(\eps - \alpha)$-small connected components in $G'$ is at least $(\eps+\alpha) m/2$.

	Let $C \subset V$ denote the set of vertices corresponding to an $(\eps - \alpha)$-small connected component in $G'$.
	If $\bigcup_{v\in C}\adj{v}$ has no erasures, 
    then $C$ is a generalized witness to disconnectedness of $G$.
	Next, assume that $\bigcup_{v\in C}\adj{v}$ has exactly one erasure.
	We show that the set $C$ is a generalized witness to disconnectedness of $G$.
	Condition~\ref{condn:1} is satisfied by definition.
	Condition~\ref{condn:2} is true since $C$ forms a connected component in $G'$. 
	To see that Condition~\ref{condn:3} holds, let $u \in C$ be the vertex with $\bot\in \adj{u}$.
	Since $C$ is a connected component in $G'$, this erased entry was completed with the label of another vertex $v \in C$. Moreover, every vertex in $C$ is reachable by a BFS from $v$, since $C$ forms a connected component in $G'$, and the erased entry is not needed for these searches because it would lead back to $v$.
	Therefore, $C$ is a generalized witness to disconnectedness of $G$ if $\bigcup_{v\in C}\adj{v}$ has exactly one erasure.

Among the $(\eps-\alpha)$-small connected components
 in $G'$, at most $\alpha m$ have at least $2$ erased entries in the union of their adjacency lists.
Hence,
the number of $(\eps - \alpha)$-small generalized witnesses to 
disconnectedness of $G$
is at least $((\eps+\alpha)m/2) - \alpha m = (\eps - \alpha)m/2$.
	\end{proof}

Lemma~\ref{lem:connectedness-tester} below implies Theorem~\ref{thm:connectedness-tester}.

\begin{lemma}\label{lem:connectedness-tester}
For every $\eps \in (0,2/\davg)$ and $\alpha \in [0,\eps),$ Algorithm~\ref{alg:connectedness} is an $\alpha$-erasure-resilient $\eps$-tester for connectedness of graphs with the average degree $\davg$. It has 
$O(b^2 \davg \cdot \min\{b/\davg, 1\})$ query and time complexity.
\end{lemma}
\begin{proof}
	Consider an $\alpha$-erased graph $G$ over a vertex set $V$.
	Assume that $G$ is connected, that is, there exists a connected completion $G'$ of $G$.
	Consider an arbitrary $C \subset V$.
	There exist vertices $u \in C$ and $v \in V\setminus C$ such that $\adj{u}$ in $G'$ contains $v$. Hence, $C$ is not a generalized witness to disconnectedness of $G$.
	Therefore, the tester accepts $G$.
	
	Next, assume that $G$ is $\eps$-far from connected.
	Let $\mathcal{W}$
	denote the family of all $(\eps - \alpha)$-small generalized witnesses to disconnectedness of $G$. 
	Let $C \subset V$ be an element of $\mathcal{W}$.
	If $\eps - \alpha \geq 4/\davg^2$, the representation length of $C$ is at most $b^2 \le b\cdot \davg$. 
	If $\eps - \alpha < 4/\davg^2$, the representation length of $C$ is at most $b\cdot \davg < b^2$. 
	Hence, the representation length of $C$ is at most $\min\{b^2, b\cdot \davg\}$.
	If $\bigcup_{v\in C}\adj{v}$ has no erasures then every vertex in $C$ is reachable from every other vertex in $C$.
	Otherwise, the vertex $v$ in Condition~\ref{condn:3} of Definition~\ref{def:witness-disconnectedness} is such a vertex.
	If Algorithm~\ref{alg:connectedness} performs a BFS from $v$, 
	it will detect a generalized witness to disconnectedness after at most $\min\{b^2, b\cdot \davg\}$ queries and reject.
	Since $|\mathcal{W}| \ge (\eps - \alpha)m/2$
	and each generalized witness in $\mathcal{W}$ has at least one vertex from which the generalized witness is detectable by a BFS,
	a single iteration of Algorithm~\ref{alg:connectedness} rejects with probability at least $|\mathcal{W}|/n = 1/b$.
	Hence, Algorithm~\ref{alg:connectedness} rejects with probability at least
	$1 - (1 - (1/b))^{\lceil b\ln 3\rceil} \geq 1 - \exp(-\ln 3) = 2/3.$
	
	Step~\ref{step:bfs} of Algorithm~\ref{alg:connectedness} makes at most $\min\{b^2, b\davg\}$ queries. 
	Thus, the query complexity of Algorithm~\ref{alg:connectedness} is $O(b\cdot \min\{b^2, b\davg\})$, 
	which simplifies to the claimed expression.
	Checking (in Step~\ref{step:witness-check}) whether a set $C$ is a generalized witness to disconnectedness can be done with a constant number of passes over the adjacency lists of vertices in $C$. Since the algorithm queried all entries in them, its running time is asymptotically equal to its 
	query complexity.
\end{proof}

\subsection{A Lower Bound for Erasure-Resilient Connectedness Testing}\label{sec:connectedness-lb}

In this section, we prove Theorem~\ref{thm:connectedness-lb}. We note that hard graphs in our construction have constant average degree. That is, for those graphs, our lower bound is $\Omega(n)=\Omega(m)$.

\begin{proof}[Proof of Theorem~\ref{thm:connectedness-lb}]
	We apply Yao's minimax principle, as stated in~\cite{RS06}.
	Specifically, we construct distributions $\Dyes$ and $\Dno$, the former over connected graphs and the latter over graphs that are $\eps$-far from connected, such that every deterministic $\eps$-erasure-resilient $\eps$-tester for connectedness makes $\Omega(m)$ queries to distinguish the two distributions.
	
	\begin{figure}
		\centering
		\includegraphics[scale=0.7]{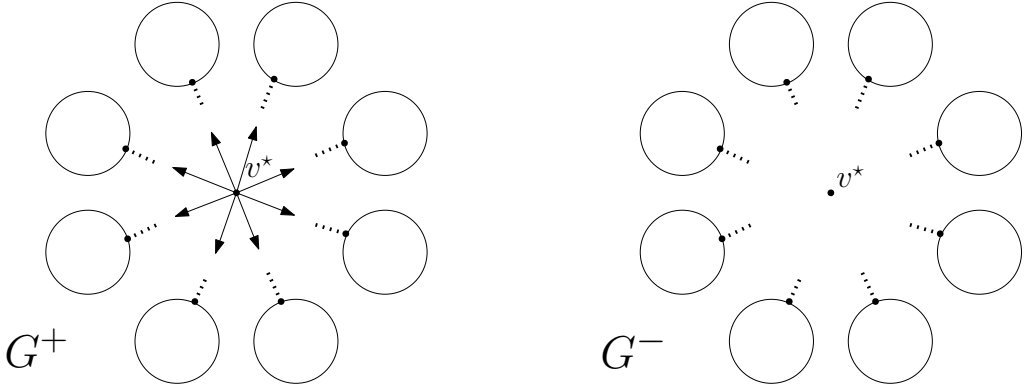}
		\caption{The partially erased graphs $\Gyes$ and $\Gno$ described in the proof of Theorem~\ref{thm:connectedness-lb}. The dotted lines represent erased entries in the adjacency lists of the corresponding vertices.
In $\Gyes$, the directed edges from $v^\star$ point to the vertices in its adjacency list.
The
circles represent cycles.}
		\label{fig:2-sided-LB}
	\end{figure}

Without loss of generality, assume that 
$t = (1-\eps)/(2\eps)$ is an integer. Observe that $t \geq 3$ as $\eps \leq 1/7$. Let $k$ be an even number and $n = kt + 1$.
We first construct two partially erased $n$-node graphs $\Gyes$ and $\Gno$, depicted in Figure~\ref{fig:2-sided-LB}.
The vertices of $\Gyes$ are partitioned into $k + 1$ sets. Each of the first $k$ sets induces a $t$-node cycle.
Exactly one node in each cycle has degree 3 and has an erasure in its adjacency list, in addition to its two neighbors on the cycle. The last set contains a single node $v^\star$ of degree $k$. Its  adjacency list contains the labels of the degree-3 vertices in the cycles. The graph $\Gno$ is the same as $\Gyes$, except that in $\Gno$, we have that $\adj{v^\star}$ is empty, that is, $v^\star$ is isolated.

	We can obtain a connected completion of $\Gyes$ by connecting the vertex $v^\star$ to all the degree-3 vertices.
	In contrast, at least $k/2$ edges need to be added to every completion of $\Gno$ to make it connected.
	Hence, the distance from $\Gno$ to connectedness is
	$(k/2)/(kt + k/2) = 1/(2t + 1) = \eps$.
	
	The fraction of erased entries in the adjacency lists of $\Gyes$ and $\Gno$ are $1/(2t + 2)$ and $1/(2t + 1)$, respectively.
	That is, $\Gyes$ and $\Gno$ are both $\alpha$-erased graphs for $\alpha = 1/(2t+1) = \eps$.
	
The distributions $\Dyes$ and $\Dno$ are uniform over the sets of all partially erased graphs isomorphic to $\Gyes$ and $\Gno$, respectively.
	Each partially erased graph sampled from $\Dyes$ is connected.
	Each partially erased graph sampled from $\Dno$ is $\eps$-far from connected.

	\begin{claim}
		Every deterministic algorithm $A$ has to make $\Omega(n)$ queries to distinguish $\Dyes$ and $\Dno$ with probability at least $2/3$.
	\end{claim}
\begin{proof}
Let $q$ denote the number of queries made by $A$ and assume $q\leq n/6$.
In this proof, we use $v^\star$ as a shorthand for the vertex from the singleton set in the construction of $\Dyes$ and $\Dno$, as opposed to the label of that vertex.
Since $\Dyes$ and $\Dno$ differ only on $v^\star$, it is important to understand when $A$ gets any information about $v^\star$.
\begin{definition}[Node status]\label{def:queried-seen-unknown}
			Given a sequence of queries made by $A$ and answers it has received so far, a node $v$ is {\em known} if it has been queried (via a degree or neighbor query) or received as an answer to a (neighbor) query; otherwise, it is \emph{unknown}.
\end{definition}

The node $v^\star$ is \emph{unknown} before $A$ makes its first query. Since $v^\star$ cannot be received as an answer to a query for the graphs in the support of $\Dyes$ and $\Dno$, it can become \emph{known} only if $A$ queries an \emph{unknown} node that happens to be $v^\star$. At most two new nodes become \emph{known} per query. So, the probability (over the distribution $\Dyes$ or $\Dno$) that a specific \emph{unknown} node queried by $A$ turns out to be $v^\star$ is at most $1/(n-2q)$.
Let $p$ denote  the probability that $v^\star$ becomes \emph{known} by the end of an execution of $A$.
By a union bound over all queries made by $A$,
\[p\leq \frac q {n-2q}\leq \frac{n/6}{n-n/3}=\frac 1 4.\]

If $v^\star$ is \emph{unknown} by the end of a particular execution then the view of the partially erased graph obtained by $A$ in that execution arises with the same probability under $\Dyes$ and under $\Dno$.
Such an execution of $A$ can distinguish $\Dyes$ and $\Dno$ with probability at most $1/2$.
		Therefore, the probability that $A$ distinguishes $\Dyes$ and $\Dno$ is at most $p+(1-p)\cdot \frac{1}{2}= \frac 1 2 +\frac p 2 < \frac{2}{3}$.
	\end{proof}
	In our construction, $m=\Theta(n)$.
	Thus, every $\eps$-erasure-resilient $\eps$-tester for connectedness that uses only degree and neighbor queries must make $\Omega(m)$ queries in the worst case over the input graph, completing the proof of Theorem~\ref{thm:connectedness-lb}.
\end{proof}

\section{Estimating the Average Degree of a Graph}\label{sec:av-degree}
In this section,
we present our results on
erasure-resilient estimation of the average degree of graphs.

\subsection{An Algorithm for Estimating the Average Degree}\label{sec:alg-av-degree}
In this section, we describe and analyze an algorithm for estimating the average degree of (or, equivalently, the number of edges in) a partially erased graph and prove Theorem~\ref{thm:degree-ub}. Our algorithm is a generalization of the algorithm for counting the number of edges in graphs by Eden et al.~\cite{ERS17, ERS} to the case of partially erased graphs.
We first give an algorithm (Algorithm~\ref{alg:new-average-degree-oracle}) that takes a crude estimate of the average degree as input and outputs a more accurate estimate. Our final algorithm (Algorithm~\ref{alg:new-average-degree-estimation}) uses Algorithm~\ref{alg:new-average-degree-oracle} as a subroutine to gradually refine its estimate of the average degree.

Algorithm~\ref{alg:new-average-degree-oracle}, like the algorithm of Eden et al.~\cite{ERS17, ERS}, works by empirically estimating a random variable whose
expectation is close to the number of edges in the graph.
We first rank vertices according to their degrees, breaking ties arbitrarily.
Then we orient the nonerased edges of the graph from lower-ranked to higher-ranked endpoints.
This orientation allows us to attribute each nonerased edge to its lower-ranked endpoint in order to avoid double-counting the edge.
Since the number of edges between high-degree vertices is small, we ignore such edges.
Algorithm~\ref{alg:new-average-degree-oracle} samples low-degree vertices uniformly at random and estimates, via sampling, the number of edges ``credited'' to them.

The crucial difference in the behavior of the algorithm in the case of partially erased graphs is the following.
When we sample an erased entry from the adjacency list of a low-degree vertex $u$, we assume that it gets completed to a vertex ranked higher than $u$ and, therefore, attribute the corresponding edge to $u$.
Consequently, some erased edges get counted twice.
This results in the additional term depending on the fraction of erasures in the approximation guarantee.

The ranking or the total ordering on the vertices of a graph is defined below.
\begin{definition}[Total ordering $\prec$]\label{def:vertex-order}
In a partially erased graph $G$, for any two vertices $u,v,$ we write $u \prec v$ if either $\deg(u) < \deg(v)$, or $\deg(u) = \deg(v)$ and $u$ is lexicographically smaller than $v$.
\end{definition}

\begin{algorithm}
	\caption{Erasure-Resilient Algorithm for Improving an Estimate~of Average~Degree}	\label{alg:new-average-degree-oracle}

	\SetKwInOut{Input}{input}\SetKwInOut{Output}{output}
	\SetKwFor{RepeatTimes}{repeat}{times}{end}
	\Input{Parameters $\eps \in (0,1/2), \delta \in (0, 1/3)$; query access to a partially erased graph $G$ on $n$ nodes; a crude estimate $\destin$ of the average degree of $G$}
	\DontPrintSemicolon
	\BlankLine

		\nl\label{step:no-samples} Set $s \gets \Big\lceil 660 \ln(2/\delta)\sqrt{\frac{n}{\eps^{5}\cdot\destin}} \Big\rceil$.\;
		\nl \For{$i = 1$ to $s$}{
		\nl Sample a node $u$ from $V$ uniformly at random and query its degree, $\deg(u)$.\;
		\nl Query the oracle for a uniformly random entry $v$ from $\adj{u}$.\;
		\nl If $v \ne \bot$ then query its degree, $\deg(v)$.\;
		\nl\label{step:deg-if}\eIf{$\deg(u) \leq 4 \sqrt{{n\destin}/{\eps}}$ \textbf{and} either $v = \bot$ or $u \prec v$}{
		\nl\label{step:chi-deg} $\chi_i \gets \deg(u)$\;}
		{
		\nl \label{step:chi-0} $\chi_i \gets 0$ \;
		}
		}
		\nl\label{step:deg-output} \Return $\dest =2 \cdot \frac{1}{s} \sum\limits_{i =1}^{s} \chi_i$\,.\;

\end{algorithm}

\begin{lemma}\label{lem:avg-deg-guarantee}
Let $G$ be an $\alpha$-erased $n$-node graph with the average degree $\davg \geq 1$. Let  $\destin$ be a crude estimate of the average degree, given as an input to  Algorithm~\ref{alg:new-average-degree-oracle}. Then the output $\dest$ of Algorithm~\ref{alg:new-average-degree-oracle} satisfies the following:
\begin{enumerate}
\item\label{enum:bad-m-est} If $\destin \geq \frac{\davg}{8}$ then, with probability at least $3/4$, we have
$\dest \leq 8\davg.$

\item\label{enum:good-m-est} Furthermore, if $\frac{\davg}{8} \leq \destin \leq 8\davg$ then with probability at least $1 - \delta$,
\[ (1-\eps) \cdot \davg < \dest < (1+\eps+2 \min(\alpha, \tfrac{1}{2})) \cdot \davg. \]
\end{enumerate}
The query complexity of the algorithm is $\Theta\left( \sqrt{\frac{n}{\eps^5 \cdot \destin}} \cdot \log\frac{1}{\delta} \right)$.
\end{lemma}

\begin{proof}
The algorithm makes at most two degree queries and one neighbor query in each iteration, and it runs for $\Theta\Big( \sqrt{\frac{n}{\eps^5 \cdot \destin}} \cdot \log\frac{1}{\delta} \Big)$ iterations. Hence, the bound on its query complexity is as claimed in the lemma.

To prove the guarantees on the output estimate $\dest$,
we first show that for all $i \in [s]$, the expected value of $\chi_i$ is a good estimate to the average degree of the partially erased graph, where $s$ is the number of samples taken by Algorithm~\ref{alg:new-average-degree-oracle}.
We then apply Markov's inequality and Chernoff bound to prove parts~\ref{enum:bad-m-est} and~\ref{enum:good-m-est} of the lemma, respectively.
For all $i \in [s]$, the random variables $\chi_i$ set by the algorithm are mutually independent and identically distributed. Hence, it suffices to bound $\E[\chi_1]$.

\begin{claim}\label{clm:chi-i}
If $\destin \geq \frac{\davg}{8}$ then 
\[\left(1 - \frac{\eps}{2}\right) \cdot \frac{\davg}{2} < \E[\chi_1] \leq  \left(1 + 2 \min\left(\alpha, \frac{1}{2}\right) \right) \cdot \frac{\davg}{2}.\]
\end{claim}

\begin{proof}
Let $m = {n\davg}/{2}$ denote the total number of edges in the graph, and
\[\cH = \left\{u \in V \ \Big| \ \deg(u) > 4\sqrt{{n\destin}/{\eps}} \right\}\]
denote the set of  high degree vertices. Let $\mestin = n\destin/2$ be the number of edges in the graph estimated from the input parameter $\destin$. Since $\destin \geq \davg/8$, we have $\mestin \geq {m}/{8}$.  Hence,
\begin{equation} \label{eq:h-bound}
|\cH| < \frac{2m}{4\sqrt{n\destin/\eps}} = \frac{m}{2\sqrt{2\mestin/\eps}} \leq \frac{m}{\sqrt{m/\eps}} = \sqrt{\eps m},
\end{equation}
where the first inequality holds because the sum of degrees of high-degree vertices is at most $2m$, and the second inequality follows from $\mestin \geq {m}/{8}$.

The following quantity, $\outdeg(u)$, was defined in~\cite{ERS} for (standard) graphs. We extend their definition to partially erased graphs.

\begin{definition}\label{def:outdegree}
For a vertex $u$ in a partially erased graph $G$, let $N(u)$ denote the set of (nonerased) neighbors present in $\adj{u}$. Let
$\outdeg(u) = |\{v \in N(u) \mid u \prec v\}|$
denote the number of nonerased neighbors of $u$ that are higher than $u$ w.r.t. the ordering on vertices (as in Definition~\ref{def:vertex-order}).
\end{definition}

Roughly, $\outdeg(u)$ denotes the number of nonerased neighbors of $u$ with the degree higher than that of $u$. The following fact is based on an observation by~\cite{ERS}.

\begin{fact}\label{fact:outdeg}
For a partially erased graph $G$ over a vertex set $V$, the sum
$ \sum_{u \in V} \outdeg(u) \leq m$.
The inequality can be replaced with equality when $G$ has no erasures.
\end{fact}

\noindent The fact holds because each nonerased and half-erased edge in $G$ is counted exactly once and at most once, respectively, in the sum $\sum_{u \in V} \outdeg(u)$.

Let $u_1, u_2, \ldots, u_{|\cH|}$ be a labeling of the the high degree vertices such that $u_1 \prec u_2 \prec \ldots \prec u_{|\cH|}$.
For each $j \in [|\cH|]$, observe that $\outdeg(u_j) \leq |\cH| - j$,
as $\outdeg(u_j)$ is at most the number of vertices that are higher than $u_j$ in the ordering.
Hence,
\begin{equation}
\sum\limits_{u \in \cH} \outdeg(u)
\leq \sum\limits_{j = 1}^{|\cH|} (|\cH| - j)
= \sum\limits_{k = 0}^{|\cH|-1} k < \frac{|\cH|^2}{2} < \frac{\eps m}{2}, \label{eq:high-deg-bound}
\end{equation}
where the last inequality follows from \eqref{eq:h-bound}.

Let $\dbot(u)$ denote the number of erased entries in $\adj{u}$.
The expectation 
\begin{align}
\E[\chi_1] = \frac{1}{n} \sum\limits_{u \in V \setminus \cH} \frac{\outdeg(u) + \dbot(u)}{\deg(u)} \cdot \deg(u) = \frac{1}{n} \sum\limits_{u \in V \setminus \cH} (\outdeg(u) + \dbot(u)) \label{eq:ex-chi-1}
\end{align}
since the degree of the sampled vertex $u$ is assigned to $\chi_1$ if and only if
\begin{enumerate}
\item $\deg(u) \leq 4 \sqrt{{n \destin}/{\eps}}$, i.e., $u \in V \setminus \cH$; and
\item the queried entry from $\adj{u}$ is either a vertex $v \succ u$ or $\bot$.
\end{enumerate}

We now bound the quantity on the right hand side of~\eqref{eq:ex-chi-1} from below and above.
Let $G'$ be an arbitrary completion of $G$, and let $\outdeg_{G'}(\cdot)$ denote the quantity defined in Definition~\ref{def:outdegree} with respect to $G'$ (instead of $G$).
For each $u \in V$, observe that $\outdeg(u) + \dbot(u) \geq \outdeg_{G'} (u)$.
Also note that the upper bound in \eqref{eq:high-deg-bound} still holds if we replace $\outdeg(\cdot)$ with $\outdeg_{G'}(\cdot)$.
Hence, from~\eqref{eq:ex-chi-1},
\begin{equation} \label{eq:chi-lb}
\E[\chi_1]
\geq \frac{1}{n} \sum\limits_{u \in V \setminus \cH} \outdeg_{G'}(u)
= \frac{1}{n} \left(m - \sum\limits_{u \in \cH} \outdeg_{G'}(u)\right)
> \left(1 - \frac{\eps}{2} \right) \frac{m}{n}.
\end{equation}
On the other hand, from~\eqref{eq:ex-chi-1},
\begin{equation} \label{eq:chi-ub-1}
\E[\chi_1] 
\leq \frac{1}{n}\sum\limits_{u \in V} (\outdeg(u) + \dbot(u))
\leq (1 + 2 \alpha) \frac{m}{n},
\end{equation}
where the last inequality uses Fact~\ref{fact:outdeg} and $\sum_{u \in V} \dbot(u) \leq 2\alpha m$.
Since $\outdeg(u) + \dbot(u) \leq \deg(u)$ for all $u \in V$, from \eqref{eq:ex-chi-1},
\begin{equation} \label{eq:chi-ub-2}
\E[\chi_1] \leq \frac{1}{n} \sum\limits_{u \in V} \deg(u) = \frac{2m}{n}.
\end{equation}
This completes the proof of Claim~\ref{clm:chi-i} because,
using \eqref{eq:chi-lb},\eqref{eq:chi-ub-1} and \eqref{eq:chi-ub-2}, we get
\[\left(1 - \frac{\eps}{2}\right) \cdot \frac{m}{n}
< \E[\chi_1]
\leq  \left(1 + 2 \min\left(\alpha, \frac{1}{2}\right) \right) \cdot \frac{m}{n}. \qedhere\]
\end{proof}

Let random variable $\chi = \frac{1}{s} \sum\nolimits_{i=1}^s \chi_i$ denote the mean of $\chi_i$'s calculated in Step~\ref{step:deg-output} of Algorithm~\ref{alg:new-average-degree-oracle}. Since all $\chi_i$'s are independent and identically distributed, $\E[\chi] = \E[\chi_1]$.
Furthermore, the output $\dest$ of the algorithm is $2\chi$ and hence,
$\E[\dest] = 2\E[\chi]$.
By Claim~\ref{clm:chi-i}, if $\destin \geq \davg/8$ then
$ \E[\dest] \leq 2\davg.$
By Markov's inequality,
$\Pr[\dest > 8\davg] \leq \Pr[\dest > 4\E[\dest]] \leq \frac{1}{4}$.
This completes the proof of part~\ref{enum:bad-m-est} of Lemma~\ref{lem:avg-deg-guarantee}.

Now consider the case when $\frac{\davg}{8} \leq \destin \leq 8\davg$.
Observe that $0 \le \chi_i \leq 4 \sqrt{{n\destin}/{\eps}}$ for all $i \in [s]$ by Step~\ref{step:deg-if}. 
Hence, by an application of the Hoeffding bound,
\begin{align*}
\Pr\left[ |\chi - \E[\chi]| \geq \frac{\eps}{2} \cdot \E[\chi] \right]
\leq 2 \exp\left(- \frac{\eps^2/4}{2 + \eps/2} \cdot \frac{s \E[\chi]}{4} \sqrt{\frac{\eps}{n\destin}} \right) < \delta,
\end{align*}
where we used $\eps < 1/2$ and $\destin \le 8\davg$ in the simplification. 
Hence, with probability at least $1-\delta$,
\[
\left(1 - \frac{\eps}{2}\right) \cdot \E[\chi_1] < \chi < \left(1 + \frac{\eps}{2}\right) \cdot \E[\chi_1].
\]
Since $\dest = 2 \chi$, by Claim~\ref{clm:chi-i}, we get that with probability at least $1-\delta$,
\begin{align*}
\left( 1 - \frac{\eps}{2} \right)\left( 1 - \frac{\eps}{2} \right)\cdot \davg & < \dest < \left( 1 + \frac{\eps}{2} \right)\left(1 + 2 \min\left(\alpha, \frac{1}{2}\right) \right) \cdot \davg,
\end{align*}
proving part~\ref{enum:good-m-est} of Lemma~\ref{lem:avg-deg-guarantee}.
\end{proof}

\begin{algorithm}
	\caption{Erasure-Resilient Algorithm for Estimating the Average Degree} \label{alg:new-average-degree-estimation}

	\SetKwInOut{Input}{input}\SetKwInOut{Output}{output}
	\SetKwFor{RepeatTimes}{repeat}{times}{end}
	\Input{Parameter $\eps \in (0,1/2)$; query access to a partially erased graph $G$ on $n$ nodes}
	\DontPrintSemicolon
	\BlankLine

		\nl Set $t \gets \ceil{12 \ln(4 \log n)}$.\;
		\nl\label{step:for-avg-degree}\For{$i = 0$ to $\ceil{\log n}$}{
		\nl Set $\destin_i \gets {n}/{2^i}$.\;
		\nl\label{step:oracle-runs}\RepeatTimes{$t$} {
		\nl Run Algorithm~\ref{alg:new-average-degree-oracle} on inputs $\eps$ and $\destin_i$ with $\delta = 1/4$.}
		\nl\label{step:median-dest} Let $\dest_i$ be the median of the answers returned by Algorithm~\ref{alg:new-average-degree-oracle} in all the runs.\;
		\nl\label{step:return-dest}\lIf{$\dest_i > \destin_i$}{\Return $\dest_i$\,.}
		}
		\nl\label{step:return-1} \Return $1$.
\end{algorithm}

\begin{proof}[Proof of Theorem~\ref{thm:degree-ub}]
Our algorithm (Algorithm~\ref{alg:new-average-degree-estimation}) uses Algorithm~\ref{alg:new-average-degree-oracle} as a subroutine.
It runs with values of initial estimates $\destin$ set in powers of 2, stopping and returning the current estimate once it exceeds the initial estimate for this iteration.

Let $\ell \in \{0,1, \ldots, \ceil{\log n}\}$ be the iteration in which the algorithm returns the estimate in Step~\ref{step:return-dest}. If the algorithm returns the estimate in Step~\ref{step:return-1} then we let $\ell$ be $\ceil{\log n} + 1$.
Consider an iteration $i \in \{0,1, \ldots, \ell\}$ of the algorithm.
Call iteration $i$ \emph{good} if $\dest_i$ satisfies the guarantees of Lemma~\ref{lem:avg-deg-guarantee} and \emph{bad} otherwise.
The probability that iteration $i$ is bad is equal to the probability that at least $t/2$ runs of Step~\ref{step:oracle-runs} fail to satisfy the guarantees of Lemma~\ref{lem:avg-deg-guarantee}. By Chernoff bound, this probability is at most $1/(4 \log n)$. Hence, by the union bound, the probability that there exists a bad iteration in the execution of the algorithm is at most $\frac{\ell+1}{4 \log n} \leq \frac{\ceil{\log n} + 1}{4 \log n}$ which is at most $1/3$ whenever $n \geq 39$.
In the rest of the proof, we condition on the event that all iterations are good.

\begin{claim}\label{clm:destin-lb}
If all iterations are good then $\destin_i \geq \davg/8$ for all $i \in \{0, 1, \ldots, \ell\}$.
\end{claim}

\begin{proof}
Since $\destin_{i-1} = 2 \destin_i$ for all $i \in [\ell]$, it suffices to prove that $\destin_\ell \geq \davg/8$. Suppose for the sake of contradiction that $\destin_\ell < \davg/8$. Then, for some iteration $k < \ell$, the estimate $\destin_k$ satisfied $\davg/4 \leq \destin_k < \davg/2$. Since iteration $k$ was good, part~\ref{enum:good-m-est} of Lemma~\ref{lem:avg-deg-guarantee} implies that $\dest_k \geq (1 - \eps)\davg > \davg/2$. Hence, $\dest_k > \destin_k$. Then Step~\ref{step:return-dest} in iteration $k$ would have returned an output and terminated the algorithm, contradicting the fact that the algorithm ran for $\ell$ iterations. Hence, $\destin_\ell \geq \davg/8$.
\end{proof}

By Step~\ref{step:return-dest}, $\destin_\ell < \dest_\ell$. By Claim~\ref{clm:destin-lb} and part~\ref{enum:bad-m-est} of Lemma~\ref{lem:avg-deg-guarantee}, the output satisfies $\dest_\ell \leq 8\davg$. Hence, $\destin_\ell \leq 8\davg$. Combining this with Claim~\ref{clm:destin-lb}, by part~\ref{enum:good-m-est} of Lemma~\ref{lem:avg-deg-guarantee}, the output of the algorithm satisfies $(1-\eps) \davg < \dest < (1+\eps+2 \min(\alpha, \tfrac{1}{2})) \davg$.

The running time of each run of Algorithm~\ref{alg:new-average-degree-oracle} in Step~\ref{step:oracle-runs} of iteration $i$ is $O\left(\frac{2^{i/2}}{\eps^{2.5}}\right)$. Furthermore, when all iterations are good, we have $n/2^\ell \geq \davg/8$ which implies that $\ell \leq \log(8n/\davg)$.
Hence, the running time of the algorithm is
\begin{align*}
O\left(\frac{\log\log n}{\eps^{2.5}} \right) \cdot \sum\limits_{i = 0}^{\ell} 2^{i/2} = O\left(\sqrt{{n}/{\davg}} \cdot \frac{\log\log n}{\eps^{2.5}} \right)
\end{align*}
when Algorithm~\ref{alg:new-average-degree-estimation} outputs the correct estimate. When it fails to output the correct estimate, the worst-case query complexity is $O\left(\sqrt{n} \cdot \frac{\log\log n}{\eps^{2.5}} \right)$.
\end{proof}

\subsection{A Lower Bound for Estimating the Average Degree}\label{sec:lb-av-degree}
In this section, we prove Theorem~\ref{thm:degree-lb}.
\begin{figure}[t]
	\centering
	\includegraphics[scale=0.9]{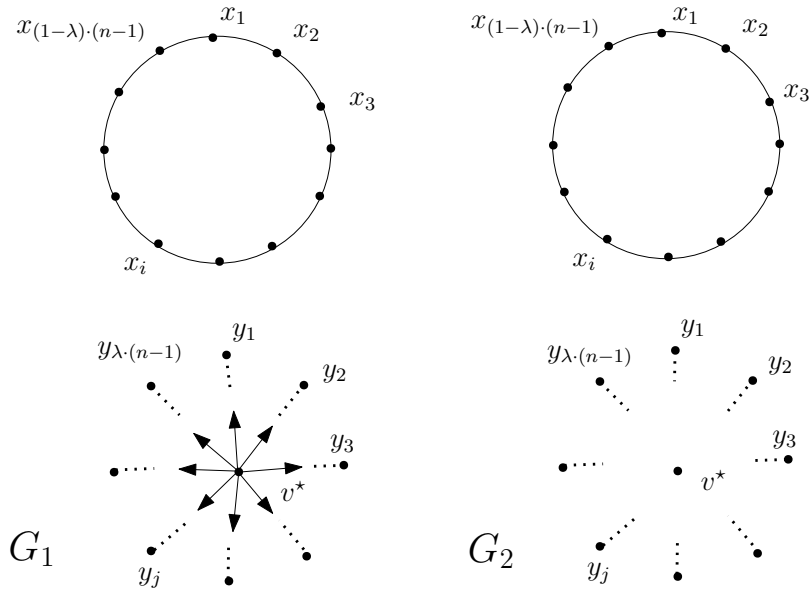}
	\caption{The partially erased graphs $G_1$ and $G_2$ described in the proof of Theorem~\ref{thm:degree-lb}. The dotted lines represent erased entries in the adjacency lists of corresponding vertices. The lines with arrows indicate that the entry corresponds to the vertex to which the arrow points to.
	The 
	circles represent the $(1-\lambda)(n-1)$-cycles.}

	\label{fig:average-degree-LB}
\end{figure}
\begin{proof}[Proof of Theorem~\ref{thm:degree-lb}]
Fix $\lambda = \frac{2\alpha}{1 + \alpha}$. Note that $\lambda \in (0, 1]$ since $\alpha \in (0,1]$.
Consider any integer $n$ such that $\lambda (n-1)$ is an even integer. Since $\alpha$ is rational, there are infinitely many such $n$.
We define two $n$-node graphs, $G_1$ and $G_2$ (see Figure~\ref{fig:average-degree-LB}).
 Both graphs contain a cycle consisting of $(1 - \lambda)(n-1)$ vertices. Of the remaining $\lambda (n - 1) + 1$ vertices, both graphs have $\lambda (n-1)$ vertices of degree $1$, with the only entry in the adjacency list of each such vertex erased.
The last vertex, called $v^\star$, is where $G_1$ and $G_2$ differ.
In $G_1$, we have that $\adj{v^\star}$ consists of
the labels of the $\lambda (n-1)$ degree-$1$ vertices.
In contrast, in $G_2$, the vertex $v^\star$ is isolated.
	
The graph $G_1$ can only be completed to a graph consisting of two components: a cycle of length $(1-\lambda)(n-1)$ and a star consisting of $\lambda (n-1)$ edges. The graph $G_2$ can only be completed to a graph consisting of a cycle of length $(1 - \lambda)(n-1)$, one isolated vertex, and a matching of size $\lambda (n-1)/2$.
Hence, the total lengths of the adjacency lists of $G_1$ and $G_2$ are $2(n-1)$ and $(2-\lambda)(n-1)$, respectively. The number of entries erased in both graphs is $\lambda (n-1)$. So, the fraction of erased entries in the adjacency lists of $G_1$ and $G_2$ are $\frac{\lambda}{2}$ and $\frac{\lambda}{2-\lambda}$, respectively. Hence, both $G_1$ and $G_2$ are $\alpha$-erased, as $\frac{\lambda}{2-\lambda} = \alpha$.
The average degree of $G_1$ and $G_2$ are $\frac{2(n-1)}{n}$ and $\frac{(2 - \lambda)(n-1)}{n}$, respectively. The ratio of the average degrees is $\frac{2}{2-\lambda} = 1 + \alpha$.
	
The rest of the proof is similar to that of Theorem~\ref{thm:connectedness-lb}. We define two distributions $\calD_1$ and $\calD_2$ as the uniform distributions over the set of all graphs isomorphic to $G_1$ and $G_2$, respectively.
	To differentiate between the two distributions, any tester must necessarily query $v^\star$ which requires $\Omega(n)$ queries.
	The ratio of the average degrees of the two distributions is $1 + \alpha$. Hence, to approximate the average degree within a factor of $(1 + \gamma),$ where $\gamma < \alpha$, any tester must query $\Omega(n)$ vertices.
\end{proof}

\section{Conclusion and Open Questions}\label{sec:open}

In this work, we initiate the study of sublinear-time algorithms for problems on partially erased graphs.
Our investigation opens up a plethora of research directions and possibilities for future work.
In what follows, we discuss several specific open questions arising from our work.

\paragraph{Phase Transitions in the Complexity of Erasure-Resilient Connectedness Testing.}
As shown in Section~\ref{sec:er-connectedness}, there is a phase transition in the complexity of connectedness testing at $\alpha = \eps$ from time independent of the size of the graph to $\Omega(n)$. Our upper bound on the complexity of this problem exhibits another, less drastic phase transition at $\alpha = \eps/2$, when the asymptotic dependence of the running time on $\eps$ and $ \alpha$  changes. We conjecture that this second phase transition is inherent (and not an artifact of our techniques).
It would be interesting to investigate whether connectedness testing when $\alpha \in [\eps/2,\eps)$ is fundamentally different from the same problem when $\alpha \in [0,\eps/2)$.

\paragraph{Erasure-Resilient Testing of Monotone Properties in the Bounded-Degree Model.}
A property of a graph is \emph{monotone}
if it is preserved under deletion
of edges and vertices. That is, if $G$ satisfies a monotone property then so does every subgraph of $G$. Many important graph properties, including bipartiteness, 3-colorability, and triangle-freeness, are monotone.

In the bounded-degree property testing model~\cite{GR02}, an $n$-node graph $G$
with the degree bound $D$ is represented as a concatenation of $n$ adjacency lists, each of length $D$.
For a vertex $v \in G$ and an index $i \in [D]$, a neighbor query $(v,i)$ returns a valid vertex in the graph if $i \leq \deg(v)$ and a special symbol, say $\aspace$\,, if $i > \deg(v)$.
The graph $G$
is $\eps$-far from satisfying a property $\cP$ if at least $\eps n D$ entries in the adjacency lists of $G$ need to be modified to make it satisfy $\cP$. 

Bounded-degree property testing can be generalized in a natural way to account for erased entries in adjacency lists.
A bounded-degree graph is $\alpha$-erased if at most $\alpha n D $ entries of its adjacency lists are erased.
We observe that a tester for a monotone property of bounded-degree graphs can be made erasure-resilient via a simple transformation.

\begin{observation}\label{obs:down-mono}
Let $\cP$ be a monotone property of graphs. Suppose there exists an $\eps$-tester for $\cP$ in the bounded-degree model that makes $q(\eps, n, D)$ queries. 
Then there exists an $\alpha$-erasure-resilient $\eps$-tester for $\cP$ in the bounded-degree model that 
makes at most $D^2 \cdot q(\eps-2\alpha, n, D)$ queries and works for all $\alpha \in(0,\eps/2)$.
\end{observation}

\begin{proof}
Fix an $\alpha$-erased bounded-degree graph $G$ on the vertex set $V$. Let $G^\star=(V,E^\star)$ be the graph consisting of only the nonerased edges of $G$ (see Definition~\ref{def:half-erased}).
We construct an oracle $\cO$ that simulates access to $G^\star$ by querying $G$.
Let $\adj{\cdot}$ and $\adjp{\cdot}$ denote the adjacency lists of $G$ and $G^\star$, respectively.
On a degree or a neighbor query for a vertex $v\in V,$ the oracle $\cO$ internally constructs
$\adjp{v}$ from $\adj{v}$ as follows:
\begin{enumerate}
\item Initialize $\adjp{v}$ to an empty list.
\item For each vertex $u \in \adj{v}$ (i.e., $u \notin \{\bot, \aspace \}$), concatenate $u$ to $\adjp{v}$ iff $v \in \adj{u}$.
\item Pad $\adjp{v}$ with special characters $\aspace$ so that its length is $D$.
\end{enumerate}
The oracle $\cO$ then answers the query with respect to the nonerased adjacency list $\adjp{v}$. As $\adj{v}$ has length at most $D$, and checking if $v \in \adj{u}$ for each $u \in \adj{v}$ takes at most $D$ queries, the oracle makes at most $D^2$ queries to $G$ to answer each query about $G^\star$.

Observe that an edge $\{u,v\} \in G^\star$ iff $u \in \adj{v}$ and $v \in \adj{u}$.
If $G$ satisfies $\cP$ then so does $G^\star$, as $G^\star$ is a subgraph of a completion of $G$ that satisfies the monotone property $\cP$.
Suppose that $G$ is $\eps$-far from satisfying $\cP$. Fix an arbitrary completion $G'$ of $G$. As $G$ is $\alpha$-erased, at most $\alpha n D$ edges of $G'$ are (fully or partially) erased in $G$. As $G^\star$ is a subgraph of $G'$ consisting of only the nonerased edges, the adjacency lists of $G$ and $G^\star$ differ on at most $2\alpha n D$ entries. As $G'$ is $\eps$-far from $\cP$, the graph $G^\star$ is $(\eps-2\alpha)$-far from $\cP$.

Let $\cT$ be an $\eps$-tester for $\cP$ whose query complexity is $q(\eps, n, D)$. Then, for $\alpha < \eps/2$, an $\alpha$-erasure-resilient $\eps$-tester $\cT'$ for $\cP$ can be obtained by simulating $\cT$ with the proximity parameter $\eps-2\alpha$ on $G^\star$ via the oracle $\cO$ and returning the result of the simulation. The complexity of $\cT'$ is $D^2 \cdot q(\eps-2\alpha, n, D)$ as the oracle $\cO$ makes at most $D^2$ queries to $G$ for each of the $q(\eps-2\alpha, n, D)$ queries it receives.
\end{proof}

This transformation is not efficient for general graphs, as the maximum degree of a graph can be $n-1$.
It is interesting to understand how much erasure-resilience affects query complexity of testing monotone properties
in our erasure-resilient model for general graphs.

\paragraph{Erasure-Resilient vs.\ Tolerant Testing of Graphs.}
For $0 \leq \eps_1 < \eps_2 < 1$, an \emph{$(\eps_1, \eps_2)$-tolerant tester} for a property $\cP$  must accept, with high probability, if the input is $\eps_1$-close\footnote{An object is $\eps_1$-close to a property $\cP$ if it is not $\eps_1$-far from $\cP$.}
to $\cP$ and reject, with high probability, if the input is $\eps_2$-far from $\cP$~\cite{PRR06}.
Dixit et al.~\cite{DRTV18} observed that, for properties of functions, erasure-resilient testing is no harder than tolerant testing.
Specifically, a tolerant tester for a property of functions can be easily converted to an erasure-resilient tester with the same complexity. The new tester can run the tolerant tester, filling in the queried erasures with arbitrary values.
However, this argument fails in the case of testing properties of graphs represented as adjacency lists, since the erased entries have to be filled in so that the resulting completion is a valid graph.
%
In the bounded-degree model,
we can use a $(2\alpha, \eps-2\alpha)$-tolerant tester for a property $\cP$ to obtain an
 $\alpha$-erasure-resilient $\eps$-tester for $\cP$ with an overhead $O(D^2)$ in query complexity via a transformation similar to the one explained in our discussion of monotone properties.
It is an important open question to understand the relationship between erasure-resilient and tolerant testing
in the general graph model.

\paragraph{Symmetric vs.\ Asymmetric Erasures.}
Our definition of partially erased graphs is general in the sense that erased entries may be \emph{asymmetric}: an edge $(u,v)$ can be erased in $\adj{u}$, but not in $\adj{v}$. A partially erased graph has only \emph{symmetric} erasures if it has no half-erased edges, 
that is, $u\in\adj{v}$ iff $v\in\adj{u}$ for any two nodes $u, v.$
It is an interesting direction to investigate which computational tasks are strictly easier in the model with symmetric erasures compared to the model with asymmetric erasures.

\bigskip

\paragraph{Acknowledgments.} We thank Talya Eden for useful discussions that led to simplification of analysis in Section~\ref{sec:alg-av-degree}.

\bibliographystyle{alpha}
\bibliography{references}

\begin{appendix}
\section{Connectedness Tester for Nonerased Graphs}\label{sec:nonerased-connectedness}

Our connectedness tester (Algorithm~\ref{alg:connectedness-work-investment}) improves on the connectedness tester of \cite{BRY14} in the standard property testing model in terms of the dependence on the proximity parameter.
Since this improvement is of independent interest, we state in this section, the special case of Algorithm~\ref{alg:connectedness-work-investment} for graphs with no erasures.
Algorithm~\ref{alg:optimal-connectedness} can be obtained by setting $\alpha=0$
in our erasure-resilient tester for connectedness for the case $b > \davg \log \avgwit$.

\begin{algorithm}
	\caption{The Connectedness Tester}\label{alg:optimal-connectedness}
	\SetKwInOut{Input}{input}\SetKwInOut{Output}{output}
	\SetKwFor{RepeatTimes}{repeat}{times}{end}
		\Input{The average degree $\davg$, parameter $\eps \in (0,{2}/{\davg})$; query access to a graph $G$}
		\DontPrintSemicolon
		\BlankLine
		
		\nl \label{step:setting-t}Set $t \gets \ceil{\log ({8}/{(\eps \davg)})}$.	\tcp*{$t = \ceil{\log(4b)}$ where $b = 2/(\eps\davg)$}	
		\nl \label{step:for-gen-conn}\For{$i \in [t] $}{
		\nl \label{step:repetition-gen-conn}\RepeatTimes{$s_i = \ceil{2^{t-i} \cdot \ln 6}$}{
		\nl Sample a vertex $v \in V$ uniformly at random and query its degree $\deg(v)$.\label{step:vertex-sample-gen-conn} \;
		\nl Run a BFS starting from $v$ until it encounters $2^{i-1} \cdot \deg(v) + 1$ edges.\label{step:bfs-gen-conn} \;
		\nl \lIf {Step~\ref{step:bfs-gen-conn} explored an entire connected component} {
		\label{step:reject-gen-conn}\textbf{reject}.}
		}
		}
		\nl \textbf{Accept}.\;
\end{algorithm}

\subsection{Connectedness Tester for Graphs with Unknown Average Degree}\label{sec:unknown-avg-degree}
We modify Algorithm~\ref{alg:optimal-connectedness} to work in the setting where the number of edges in the input graph (and, consequently, its average degree $\davg$) is not known to the algorithm. We present the modification in Algorithm~\ref{alg:optimal-connectedness-unknown-davg}.

\begin{algorithm}
	\caption{The Connectedness Tester for Graphs with unknown Average Degree}\label{alg:optimal-connectedness-unknown-davg}
	\SetKwInOut{Input}{input}\SetKwInOut{Output}{output}
	\SetKwFor{RepeatTimes}{repeat}{times}{end}
	\SetKwFor{Loop}{loop}{}{end}
		\Input{Parameter $\eps \in (0,1)$; query access to a graph $G$} 
		\setcounter{AlgoLine}{-1}
		\DontPrintSemicolon
		\BlankLine
		\nl Abort and \textbf{accept} if the number of neighbor queries exceeds $\frac{350}{\eps} \log \frac{16}{\eps}$. \;
		\nl Initialize $t \gets 0$.\;
		\nl \Loop{}{
		\nl Increment $t \gets t + 1$. \;
		\nl \For{$i \in [t]$}{
		\nl \RepeatTimes{$\ceil{2^{\max\{t-i-1, 0\}} \cdot \ln 6}$}{
		\nl Sample a vertex $v \in V$ uniformly at random and query its degree $\deg(v)$.\;
		\nl Run a BFS starting from $v$ until it encounters $2^{i-1} \cdot \deg(v) + 1$ edges.\label{step:bfs-nonerased-conn-1} \;
		\nl \lIf {Step~\ref{step:bfs-nonerased-conn-1} explored an entire connected component} {\textbf{reject}.\label{step:rej-nonerased-conn-1}}
		}
		}
		}
\end{algorithm}

\begin{theorem} \label{thm:connectedness-without-davg}
Algorithm~\ref{alg:optimal-connectedness-unknown-davg} is an $\eps$-tester for connectedness that has $O(\frac{1}{\eps} \log \frac{1}{\eps})$ query and time complexity and has 1-sided error.
\end{theorem}

We exploit the fact that $\davg$ is only used in Step~\ref{step:setting-t} of Algorithm~\ref{alg:optimal-connectedness} to set the value of $t$.
When $\davg$ is unknown, we try different values of $t$, starting at $t = 1$ and increasing it in each iteration until the tester rejects, or the query budget is reached.
Importantly, when $t = \ceil{\log({8}/{(\eps\davg)})}$,
Algorithm~\ref{alg:optimal-connectedness-unknown-davg} is identical to Algorithm~\ref{alg:optimal-connectedness}.
Clearly, Algorithm~\ref{alg:optimal-connectedness-unknown-davg} never rejects connected graphs.
For graphs that are $\eps$-far from connected, if the algorithm rejects in Step~\ref{step:rej-nonerased-conn-1} 
before the value of $t$ reaches $\ceil{\log({8}/{(\eps\davg)})}$, we are done.
The expected cumulative number of queries made by the algorithm for $t$ to be at least $\ceil{\log({8}/{(\eps\davg)})} + 1$ is $O\big( \frac{1}{\eps} \log\frac{1}{\eps \davg}\big)$.
Hence, by Markov's inequality, with high constant probability, 
Algorithm~\ref{alg:optimal-connectedness-unknown-davg} does not exceed its query budget until $t$ is at least $\ceil{\log({8}/{(\eps\davg)})}$.
The rest of the correctness argument is identical to that of Algorithm~\ref{alg:optimal-connectedness}.
\end{appendix}

\end{document}